\documentclass{mscs}
\usepackage{authblk}
\usepackage{framed}
\usepackage{url}
\usepackage{mdframed}
\usepackage{tcolorbox}
\usepackage{algorithm}
\usepackage{algorithmicx}
\let\proof\relax

\usepackage{amsmath, amsfonts,amssymb,amsthm}
\newtheorem{example}{Example}

\newtheorem{theorem}{Theorem}
\newtheorem{lemma}[theorem]{Lemma}
\newtheorem{definition}{Definition}
\newcommand{\ignore}[1]{{}}

\newcommand{\blue}[1]{{\color{blue} #1}}

\newcommand{\Iach}{\mathfrak{I}_{ACh}}
\newcommand{\Var}{\mathit{Var}}
\newcommand{\calT}{\mathcal{T}}
\newcommand{\calF}{\mathcal{F}}
\newcommand{\calX}{\mathcal{X}}
\newcommand{\calV}{\mathcal{V}}
\newcommand{\striple}{\Gamma||\triangle||\sigma}
\newcommand{\striplep}{\Gamma'||\triangle'||\sigma'}

\newcommand{\mstriple}{\mathbb{M}_{\Iach}(\Gamma, \triangle, \sigma)}
\newcommand{\mstriplep}{\mathbb{M}_{\Iach}(\Gamma', \triangle', \sigma')}

\title[Bounded ACh Unification]{Bounded ACh Unification}
\author[Ajay K. Eeralla]{A\ls J\ls A\ls Y\ns K\ls U\ls M\ls A\ls R\ns E\ls E\ls R\ls A\ls L\ls L\ls A$^1$\thanks{Ajay Kumar Eeralla was partially supported by NSF Grant CNS 1314338}\ns and\ns C\ls H\ls R\ls I\ls S\ls T\ls O\ls P\ls H\ls E\ls R\ns L\ls Y\ls N\ls C\ls H$^2$ \addressbreak \addressbreak
   $^1$  Department of Electrical Engineering and Computer Science, University of Missouri, \addressbreak \hspace{0.25cm}Columbia, USA.
     \addressbreak $^2$ Department of Computer Science, Clarkson University, Potsdam, USA }

\begin{document}

\maketitle
\begin{abstract}
We consider the problem of the unification modulo an equational theory ACh, which consists of a function symbol $h$ that is homomorphic over an associative-commutative operator $+$. Since the unification modulo ACh theory is undecidable, we define a variant of the problem called \textit{bounded ACh unification}. In this bounded version of ACh unification, we essentially bound the number of times $h$ can be applied to a term recursively, and only allow solutions that satisfy this bound. There is no bound on the number of occurrences of $h$ in a term, and the $+$ symbol can be applied an unlimited number of times. We give inference rules for solving the bounded version of the problem and prove that the rules are sound, complete, and terminating. We have implemented the algorithm in Maude and give experimental results. We argue that this algorithm is useful in cryptographic protocol analysis.
 \end{abstract}
 \section{Introduction}
Unification is a method to find a solution for a set of equations. For instance, consider an equation $x + y \overset{?}= a + b$, where $x$ and $y$ are variables, and
$a$, and $b$ are constants. If $+$ is an uninterpreted function symbol, then the equation has one solution $\{ x \mapsto a,\,y \mapsto b\}$, and this unification is called syntactic unification. If the function symbol $+$ has the property of commutativity then the equation has two solutions: $\{ x \mapsto a,\,y \mapsto b\}$ and $\{ x \mapsto b,\,y \mapsto a\}$;
And this is called unification modulo the commutativity theory.\par
Unification modulo equational theories play a significant role in symbolic cryptographic protocol analysis~\cite{EMM}.
An overview and references for some of the algorithms may be seen in~\cite{KNW, EKLMMNS, NMM}.
One such equational theory is the distributive axioms:
$x \times (y + z) = (x \times y) + (x \times z);
(y + z)\times x = (y \times x) + (z \times x).$
A decision algorithm is presented for unification modulo two-sided distributivity in~\cite{MS}. A sub-problem of this, unification modulo one-sided distributivity, is in greater interest since many cryptographic protocol algorithms satisfy
the one-sided distributivity. In their paper~\cite{TA}, Tiden and Arnborg presented an algorithm for unification modulo one-sided distributivity: $x \times (y + z) = (x \times y) + (x \times z),$ and also it has been shown that it is undecidable if we add the properties of associativity $x + (y + z) = (x+y)+z$ and a one-sided unit element $x\times1=x$.
However, some counter examples~\cite{NMM} have been presented showing that the complexity of the algorithm is exponential, although they thought it was polynomial-time bounded.
\par
 For practical purposes, one-sided distributivity can be viewed as the homomorphism theory, $h(x+y) = h(x) +h(y)$, where the unary operator $h$ distributes over the binary operator $+$. Homomorphisms are highly used in cryptographic protocol analysis. In fact, Homomorphism is a common property
that many election voting protocols satisfy~\cite{KRS}. \par
Our goal is to present a novel construction of an algorithm to solve unification modulo the homomorphism theory over a binary symbol $+$ that also has the properties of associativity and commutativity (ACh), which is an undecidable unification problem~\cite{PN}. Given that ACh unification is undecidable but necessary to analyze cryptographic protocols, we developed an approximation of ACh unification, which we show to be decidable.
\par
In this paper, we present an algorithm to solve a modified general unification problem modulo the ACh theory, which we call {\em bounded ACh unification}. We define the {\em h-height} of a term to be basically the number of 
$h$ symbols recursively applied to each other. We then only search for ACh unifiers of a bounded h-height. We do not restrict the h-height of terms in unification problems. Moreover, the number of occurrences of the + symbol is bounded neither in a problem nor in its solutions.
In order to accomplish this, we define the {\em h-depth} of a variable, which is the number of $h$ symbols on top of a variable. We develop a set of inference rules for ACh unification that keep track of the h-depth of variables.
If the h-depth of any variable exceeds the bound $\kappa$, then the algorithm terminates with no solution. Otherwise, it gives all the unifiers or solutions to the problem. 
\section{Preliminary}
 \ignore{We assume that the reader is familiar with the basic notation of unification theory and term rewriting systems (see for example~\cite{BN, BS}).\\ \\}
 \subsection{Basic Notation}
 We briefly recall the standard notation of unification theory and term rewriting systems from~\cite{BN, BS}. \par
 Given a finite or countably infinite set of function symbols $\calF$, also known as a signature, and a countable set of variables $\mathcal{V}$, the set of $\calF$-terms over $\mathcal{V}$ is denoted by $\calT( \calF, \mathcal{V})$. The set of variables appearing in a term $t$ is denoted by $Var(t)$, and it is extended to sets of equations. A term is called ground if $Var(t)=\emptyset$.
 Let $Pos(t)$ be the set of positions of a term $t$ including the root position $\epsilon$~\cite{BS}. For any $p\in Pos(t)$, $t|_p$ is the subterm of $t$ at the position $p$ and $t[s]_p$ is the term
$t$ in which $t|_p$ is replaced by $s$.
\ignore{\blue{For any position~$p$ in a term $t$ (including the root position $\epsilon$), $t(p)$ is the symbol at position $p$,
$t|_p$ is the subterm of $t$ at position $p$, and $t[u]_p$ is the term
$t$ in which $t|_p$ is replaced by $u$.}}A substitution is a mapping from $\calV$ to $\calT( \calF, \mathcal{V})$ with only finitely many variables not mapped to themselves and is denoted by $\sigma = \{x_1 \mapsto t_1,\ldots, x_n \mapsto t_n\}$, where the domain of $\sigma$ is $Dom(\sigma):=\{x_1,\ldots, x_n\}$. The range of $\sigma$, denoted as $Range(\sigma)$, defined as union of the sets $\{x\sigma\}$, where $x$ is a variable in $Dom(\sigma)$. The identity substitution is a substitution that maps all the variables to themselves. 
The application of substitution $\sigma$ to a term $t$, denoted as $t\sigma$, is defined by induction on the structure of the terms:
\begin{itemize}
\item $x\sigma$, where $t$ is a variable $x$
\item $c$, where $t$ is a constant symbol $c$
\item $f(t_1\sigma, \ldots,t_n\sigma)$, where $t=f(t_{1}, \ldots, t_{n})$ with $n \geq 1$
\end{itemize}
\ignore{$$t\sigma = \left\{ \begin{array}{lcl}
x\sigma & \mbox{if}& t = x, \\ 
f(t_1\sigma, \ldots,t_n\sigma) & \mbox{if} & t=f( t_1, \ldots,t_n)
\end{array}
\right.$$
In the second case of this definition, $n=0$ is allowed: in this case, $f$ is a constant symbol and we have $f\sigma = f$. 
An application of a substitution $\sigma$ to a term $t$ is denoted as $t\sigma$ and defined as, $x\sigma$, $f(t_1\sigma,\ldots, t_n\sigma)$ when $t$ is a variable $x$ or otherwise. Of course, an application on constant symbol gives the constant back as a result.}
\par
 The restriction of a substitution $\sigma$ to a set variables $\calV$, denoted as $\sigma|\calV$, is the substitution which is equal to identity everywhere except over $\calV \cap Dom(\sigma)$, where it is coincides with $\sigma$.
  
 \begin{definition}[More General Substitution]
\emph{A substitution $\sigma$ is more general than substitution $\theta$ if there exists a substitution $\eta$ such that $\theta = \sigma \eta$, denoted as $\sigma \lesssim \theta$. Note that the relation $\lesssim$ is a quasi-ordering, i.e., reflexive and transitive.}
\end{definition}
\begin{definition}[Unifier, Most General Unifier]
\emph{A substitution $\sigma$ is a unifier or solution of two terms $s$ and $t$ if $s\sigma = t\sigma$; it is a most general unifier if for every unifier $\theta$ of $s$ and $t$, $\sigma \lesssim \theta$.
Moreover, a substitution $\sigma$ is a solution of a set of equations if it is a solution of each of the equations. If a substitution $\sigma$ is a solution of a set of equations $\Gamma$, then it is denoted by
$\sigma \models \Gamma$.}
\end{definition}
A set of identities $E$ is a subset of $\calT( \calF, \mathcal{V}) \times \calT( \calF, \mathcal{V})$ and are represented in the form 
$s\thickapprox t$. An equational theory $=_E$ is induced by a set of fixed identities $E$ and it is the least congruence relation that is closed under substitution and contains $E$.
\begin{definition}[\textit{E}-Unification Problem, \textit{E}-Unifier, \textit{E}-Unifiable]
\emph{Let $\calF$ be a signature and $E$ be an equational theory.
An \textit{E}-unification problem over $\calF$ is a finite set of equations $ \Gamma = \{ s_1 \overset{?}=_E t_1,\ldots, s_n \overset{?}=_E t_n\}$ between terms. An \textit{E}-unifier or E-solution of two terms $s$ and $t$ is a substitution $\sigma$ such that $s\sigma =_E t\sigma$.
 An \textit{E}-unifier of $\Gamma$ is a substitution $\sigma$ such that $s_i\sigma =_E t_i \sigma$ for 
 $i= 1,\ldots,n$. The set of all \textit{E}-unifiers is denoted by $\mathcal{U}_E(\Gamma)$ and $\Gamma$ is called \textit{E}-unifiable
 if $\mathcal{U}_E(\Gamma) \neq \emptyset$. If $E = \emptyset$ then $\Gamma$ is a syntactic unification problem.}
\end{definition}
Let $\Gamma = \{ s_1 \overset{?}=_E t_1,\ldots,s_n \overset{?}=_E t_n\}$ be a set of equations, and let $\theta$ be a substitution.
\ignore{We say that $\theta$ satisfies $\Gamma$ modulo equational theory $E$ if $\theta$ is an \textit{E}-solution of each equation in $\Gamma$, that is,
 $ s_i \theta =_E t_i\theta$ for $i=1,\ldots,n$.} We write $\theta \models_E \Gamma$ when $\theta$ is an \textit{E}-unifier of $\Gamma$.
Let $\sigma = \{ x_1 \mapsto t_1,\ldots,x_n \mapsto t_n\}$ and $\theta$ be substitutions, and let $E$ be an equational theory.
We say that $\theta$ satisfies $\sigma$ in the equational theory $E$ if $x_i \theta =_E t_i \theta $ for $i=1,\ldots,n$. We write it as $\theta \models_E \sigma$.

\begin{definition}\emph{ Let $E$ be an equational theory and $\calX$ be a set of variables. The substitution $\sigma$ is \textit{more general modulo E on $\calX$} than $\theta$ iff there exists a substitution $\sigma'$ such that $x\theta =_E x \sigma \sigma'$ for all $x \in \calX$. We write it as $\sigma\lesssim^\calX_E \theta$.}
\end{definition}

\begin{definition}[Complete Set of \textit{E}-Unifiers]
\emph{Let $\Gamma$ be a \textit{E}-unification problem over $\calF$ and let $Var(\Gamma)$ be the set of all variables occurring in $\Gamma$. A complete set of \textit{E}-unifiers of $\Gamma$ is a set $S$ of substitutions such that, each element of $S$ is a \textit{E}-unifier of $\Gamma$, i.e., $S \subseteq \mathcal{U}_E(\Gamma)$, and for each $\theta \in \mathcal{U}_E(\Gamma)$ there exists a $\sigma \in S$ such that $\sigma$ is more general modulo E on $Var(\Gamma)$ than $\theta$, i.e., $\sigma \lesssim^{Var(\Gamma)} _E\theta$.}
\end{definition}

A complete set $S$ of \textit{E}-unifiers is minimal if for any two distinct unifiers $\sigma$ and $\theta$ in $S$, one is not more general modulo $E$ than the other,
i.e., $\sigma \lesssim^{Var(\Gamma)}_E \theta $ implies $\sigma = \theta$.
A minimal complete set of unifiers for a syntactic unification problem $\Gamma$ has only one element if it is not empty. It is denoted by {$mgu(\Gamma)$} and can be called most general unifier of unification problem $\Gamma$.
\begin{definition}
\emph{Let E be an equational theory. We say that a multi-set of equations $\Gamma'$ is a conservative E-extension of another multi-set of equations $\Gamma$ if any solution of $\Gamma'$ is also a solution of $\Gamma$ and any solution of $\Gamma$ can be extended to a solution of $\Gamma'$. This means for any solution $\sigma$ of $\Gamma$, there exists $\theta$ whose domain is the variables in $Var(\Gamma') \setminus Var(\Gamma)$ such that $\sigma\theta$ is a solution of $\Gamma$. The property of conservative E-extension is transitive.}
\end{definition}
\par
Let $\calF$ be a signature, and $l,r$ be $\calF$-terms. A \emph{rewrite rule} is an identity, denoted as $l \rightarrow r$, where $l$ is
not a variable and $\Var(r) \subseteq \Var(l)$.
A \emph{term rewriting system} (TRS) is a pair $(\calF, R)$, where $R$ is a finite set of rewrite rules. In general, a TRS is represented by $R$. A term $u$ rewrites to a term $v$ with respect to $R$, denoted
by $u \rightarrow_R v$ (or simply $u \rightarrow v$), if there exist a
position $p$ of $u$, $l \rightarrow r \in R$, and substitution
$\sigma$ such that $u|_p = l\sigma$ and $v = u[r\sigma]_p$.
A TRS $R$ is said to be \emph{terminating} if there is no infinite reduction sequences of the form $u_0 \rightarrow_R u_1 \rightarrow_R \ldots $. A TRS $R$
is \emph{confluent} if, whenever $u \rightarrow_R^{*} s_1$ and $u \rightarrow_R^{*} s_2$, there exists a term $v$ such that $s_1 \rightarrow_R^{*} v$ and
$s_2 \rightarrow_R^{*} v$. A TRS $R$ is convergent if it is both confluent and terminating.

\subsection{ACh Theory}
The equational theory we consider is the theory of a homomorphism over a binary function symbol $+$ which satisfies the properties of associativity and the commutativity. We abbreviate this theory as ACh. The signature $\calF$ includes a unary symbol $h$, and a binary symbol
$+$, and other uninterpreted function symbols with fixed arity. 
\par
The function symbols $h$ and $+$ in the signature $\calF$ satisfy the following identities:
\begin{itemize}
\item$x + (y+z) \thickapprox (x+y)+z \text{ (Associativity, A for short)}$
 \item $x + y \thickapprox y+x \text{ (Commutativity, C for short)}$
 \item $h(x + y) \thickapprox h(x) + h(y) \text{ (Homomorphism, h for short)}$
 \end{itemize}
 \subsection{Rewriting Systems}
We consider two convergent rewriting systems $R_1$ and $R_2$ for homomorphism $h$ modulo associativity and commutativity.
\begin{itemize}
\item $R_1 := \{h(x_1+x_2) \rightarrow h(x_1) + h(x_2)\}$ and
 \item $R_2 := \{h(x_1)+h(x_2) \rightarrow h(x_1+ x_2)\}$.
 \end{itemize}

 \subsection{h-Depth Set}
 \label{sec:hdset}
 For convenience, we assume that our unification problem is in {\em flattened} form, i.e., that every equation in the problem is in one of the following forms: $x \overset{?} = y$, $x \overset{?}= h(y)$, $x \overset{?}= y_1 + \cdots + y_n$, and $x \overset{?}= f ( x_1,\ldots, x_n)$, where $x$ and $y$ are variables, $y_i$s and $x_i$s are pair-wise distinct variables,
 and $f$ is a free symbol with $n\geq0$. The first kind of equations are called {\em VarVar equations}.
 The second kind are called {\em $h$-equations}. The third kind are called 
 {\em $+$-equations}. The fourth kind are called {\em free equations}.
 \ignore{It is well-known that how to convert a unification problem into flattened form.}
\begin{definition}[Graph $\mathbb{G}(\Gamma)$]
\emph{Let $\Gamma$ be a unification problem. We define a graph $\mathbb{G}(\Gamma)$ as a graph where
each node represents a variable in $\Gamma$ and each edge represents a function symbol in $\Gamma$. To be exact, if an equation $y\overset{?}=f(x_1,\ldots, x_n)$, where $f$ is a symbol with $n\geq1$, is in $\Gamma$ then the graph $\mathbb{G}(\Gamma)$ contains $n$ edges $y\overset{f}\rightarrow x_1,\ldots,y\overset{f}\rightarrow x_n$.
For a constant symbol $c$, if an equation $y\overset{?}= c$ is in $\Gamma$ then the graph $\mathbb{G}(\Gamma)$ contains a vertex $y$.
Finally, the graph $\mathbb{G}(\Gamma)$ contains two vertices $y$ and $x$ if an equation $y\overset{?} = x$ is in $\Gamma$.}
\end{definition}
\begin{definition}[h-Depth]
\emph{Let $\Gamma$ be a unification problem and let $x$ be a variable that occurs in $\Gamma$. Let $h$ be a unary symbol and let $f$ be a symbol (distinct from $h$) with arity greater than or equal to 1 and occurring in $\Gamma$.
We define h-depth of a variable $x$ as the maximum number of $h$-symbols along a path to $x$ in $\mathbb{G}(\Gamma)$, and it is denoted by $h_d(x, \Gamma)$. That is, $$h_d(x, \Gamma) := \max \{h_{dh}(x, \Gamma), h_{df}(x, \Gamma), 0 \}, $$ where $h_{dh}(x, \Gamma) := \max\{ 1+ h_d(y, \Gamma) \mid y\overset{h} \rightarrow x \text{ is an edge in } \mathbb{G}(\Gamma)\}$ and $h_{df}(x, \Gamma) := \max\{ h_d(y, \Gamma) \mid \text{there exists $f \neq h$ such that }y\overset{f}\rightarrow x \text{ is in } \mathbb{G}(\Gamma) \}$}.
 \end{definition}
\begin{definition}[h-Height]
\ignore{Let $\Gamma$ be a unification problem and let $t$ be a term that occurs in $\Gamma$.}
\emph{We define h-height of a term $t$ as the following:
$$h_h(t) := \left\{ \begin{array}{lcl}
h_h(t') + 1& \mbox{if}& t = h(t') \\ 
\max\{h_h(t_1),\ldots, h_h(t_n) \} & \mbox{if} & t = f (t_1,\ldots,t_n), f\neq h\\
 0 & \mbox{if}& t= x \,\text{or}\, c
\end{array}
\right.$$
where $f$ is a function symbol with arity greater than or equal to 1.}
 \end{definition}
 \ignore{Without loss of generality, we assume that h-depth and h-height is not defined for a variable that occurs on both sides of the equation. This is because the occur check rule--- which concludes the problem with no solution---presented in the next section has higher priority over the h-depth updating rules.}
 
\begin{definition}[h-Depth Set]
\emph{Let $\Gamma$ be a set of equations. The h-depth set of $\Gamma$, denoted $h_{ds}(\Gamma)$, is defined as
$h_{ds}(\Gamma) := \{ (x, h_{d}(x, \Gamma)) \mid x \text{ is a variable appearing in } \Gamma\}$.
\ignore{Let $\mathcal{V}$ be a set of variables occurring in $\Gamma$. We define a set h-depth set of $\Gamma$ whose elements are 
pairs of a variable from $\mathcal{V}$ and a non-negative integer.}In other words, the elements in the h-depth set are of the form $(x,\,c)$, where $x$ is a variable that occur in $\Gamma$ and
$c$ is a natural number representing the h-depth of $x$.}
\end{definition}
Maximum value of h-depth set $\triangle$ is the maximum of all $c$ values and it is denoted by $MaxVal(\triangle)$, i.e.,
$MaxVal(\triangle) := \max\{c \,|\, (x,\, c ) \in \triangle \text{ for some } x\}.$
\begin{definition}[\textit{ACh}-Unification Problem, Bounded \textit{ACh}-Unifier]
 \emph{An \textit{ACh}-unification problem over $\calF$ is a finite set of equations 
$ \Gamma = \{ s_1 \overset{?}=_{ACh} t_1,\ldots,s_n \overset{?}=_{ACh} t_n\}, s_i , t_i \in \calT( \calF, \mathcal{V}),$ where $ACh$ is the equational theory defined above.
A $\kappa$ bounded \textit{ACh}-unifier or $ \kappa$ bounded \textit{ACh}-solution of $\Gamma$ is a substitution $\sigma$ such that $s_i\sigma =_{ACh} t_i \sigma$, $h_h(s_i \sigma) \leq \kappa$, and $h_h( t_i\sigma) \leq \kappa$ for all $i$.}
\par
\emph{Notice that the bound $\kappa$ has no role in the problem but in the solution.}
\end{definition}
\setlength{\fboxsep}{1pt}
\section{Inference System $\mathfrak{I}_{ACh}$}
\subsection{Problem Format}
An inference system is a set of inference rules that transforms an equational unification problem into other.
In our inference procedure, we use a set triple $ {\Gamma||\triangle||\sigma}$ similar to the format presented in~\cite{LL}, where $\Gamma$ is a unification problem modulo the ACh theory, $\triangle$ is an h-depth set of $\Gamma$, and $\sigma$ is a substitution.
Let $\kappa \in \mathbb{N}$ be a bound on the h-depth of the variables. 
A substitution $\theta$ satisfies the set triple $ {\Gamma||\triangle||\sigma}$ if
$\theta$ satisfies $\sigma$ and every equation in $\Gamma$, $MaxVal(\triangle) \leq \kappa$, and we write that relation as $\theta \models {\Gamma||\triangle||\sigma}$.
We also use a special set triple $\bot$ for no solution in the inference procedure. Generally, the inference procedure is
based on the priority of rules and also uses \textit{don't care} non-determinism when there is no priority. i.e., any rule applied from a set
of rules without priority. Initially, $\Gamma$ is the non-empty set of equations to solve, $\triangle$ is an empty set, and $\sigma$ is the identity substitution. The inference rules are applied until either the set of equations is empty with most general unifier $\sigma$ or $\bot$ for no solution. Of course, the substitution $\sigma$ is a $\kappa$ bounded \textit{E}-unifier of $\Gamma$.
An inference rule is written as $ \frac{\Gamma||\triangle|| \sigma}{\Gamma'||\triangle'|| \sigma'}.$
This means that if something matches the top of this rule, then it is to be replaced with 
the bottom of the rule. \ignore{In the proofs we will write inference rules as follows:
$$\Gamma ||\triangle||\sigma \Rightarrow_{\Iach} \{\Gamma_1 ||\triangle_1||\sigma_1, \cdots, \Gamma_n ||\triangle_n|\sigma_n\}$$
meaning to branch and replace the left hand side with one of the right hand sides in each branch.
The only inference rule that has more than one branch is {\em AC Unification}.
So we often just write inference rules as follows:
$\Gamma ||\triangle||\sigma \Rightarrow_{\Iach} \Gamma' ||\triangle'||\sigma'$.}
\par
 Let $\mathcal{OV}$ be the set of variables occurring in the unification problem $\Gamma$ and let $\mathcal{NV}$ be a new set of variables such that $\mathcal{NV}=\mathcal{V}\setminus\mathcal{OV}$.
Unless otherwise stated we assume that $x,\,x_1,\ldots,x_n,\,\text{and}\,y,\,y_1,\ldots,y_n,\,z$ are variables in $\mathcal{V}$, $v,\,v_1,\ldots,v_n$ are in $\mathcal{NV}$, 
and terms $w, t,\,t_1,\ldots,t_n,\,s,\,s_1,\ldots,s_n$ in $\calT( \calF, \mathcal{V})$, and $f$ and $g$ are uninterpreted function symbols.
\ignore{Recall that $h$ is a unary, and the associativity and the commutativity operator $+$. }A fresh variable is a variable that is generated by the current inference rule and has never been used before.
\par
For convenience, we assume that that every equation in the problem is in one of the flattened forms (see Section~\ref{sec:hdset}).
If not, we apply flattening rules to put the equations into that form. These rules are performed before any other inference rule. They put the problem into flattened form and all the other inference rules leave the problem in flattened form, so there is no need to perform these rules again later. It is necessary to update the h-depth set $\triangle$ with the h-depth values for each variable during the inference procedure.
\subsection{Inference Rules}
\label{sec:inf:system}
We present a set of inference rules to solve a unification problem modulo associativity, commutativity, and homomorphism theory.\ignore{our inference procedure looking for the solutions within the given bound $\kappa$.} We also present some examples that illustrate the applicability of these rules.
\subsubsection{Flattening}
\hfill\par
Firstly, we present a set of inference rules for flattening the given set of equations. The variable $v$ represents a fresh variable in the following rules. 
\\ \\
\fbox{
\begin{minipage}{\textwidth}
\smallskip
\textbf{Flatten Both Sides (FBS)}$$ \frac{\{ t_1 \overset{?}= t_2\}\,\cup \,\Gamma||\triangle|| \sigma }{\{ v \overset{?}= t_1,\,v \overset{?}= t_2\}\,\cup \,\Gamma||\{(v,\, 0)\}\cup\triangle|| \sigma } 
\text{\hspace{1pt} if $t_1$ and $t_2$ $\notin\mathcal{V}$}$$ 
\smallskip
\end{minipage}
}
\\
\fbox{
\begin{minipage}{\textwidth}
\smallskip
\textbf{Flatten Left $+$ (FL)}
$$\frac{\{ t\overset{?}= t_1 + t_2\}\,\cup \,\Gamma||\triangle|| \sigma }{\{ t\overset{?}= v + t_2,\,v \overset{?}= t_1 \}\,\cup \,\Gamma||\{(v,\, 0)\}\cup\triangle|| \sigma } 
\text{\hspace{35pt} if $t_1 \notin \mathcal{V}$}$$
\smallskip
\end{minipage}
}\\ \\
\fbox{
\begin{minipage}{\textwidth}
\smallskip
\textbf{Flatten Right $+$ (FR)}
$$\frac{\{ t\overset{?}= t_1 + t_2\}\,\cup \,\Gamma||\triangle|| \sigma }{\{ t\overset{?}= t_1 + v,\,v \overset{?}= t_2\}\,\cup \,\Gamma||\{(v,\, 0)\}\cup\triangle|| \sigma } 
\text{\hspace{35pt} if $t_2 \notin \mathcal{V}$}$$
\smallskip
\end{minipage}
} \\ \\
\fbox{
\begin{minipage}{\textwidth}
\smallskip
\textbf{Flatten Under $h$ (FU)}
$$\frac{\{ t_1 \overset{?}= h(t)\}\,\cup \,\Gamma||\triangle|| \sigma }{\{ t_1\overset{?}= h(v),\,v \overset{?}= t \}\,\cup \,\Gamma||\{(v,\, 0)\}\cup\triangle|| \sigma } 
\text{\hspace{35pt} if $t \notin \mathcal{V}$}$$
\smallskip
\end{minipage}
} \\ 
\par
We demonstrate the applicability of these rules using the example below.
\begin{example}
\label{ex:flatten}
\emph{Solve the unification problem $\{h(h(x)) \overset{?}= (s + w)+(y+z) \}.$\\ \\
We only consider the set of equations $\Gamma$ here, not the full triple.\\ \\
 $\{h(h(x)) \overset{?}= (s + w)+(y+z) \}\overset{FBS}\Rightarrow\\
 \{ v \overset{?}= h(h(x)),\, v \overset{?}= (s + w)+(y+z) \} \overset{FL}\Rightarrow\\
 \{v \overset{?}= h(h(x)),\,v \overset{?}= v_1 +(y+z),\, v_1 \overset{?}= s+w\} \overset{FL}\Rightarrow\\
 \{v \overset{?}= h(h(x)),\,v \overset{?}= v_1 +(y+z),\, v_1 \overset{?}= v_2+w, v_2 \overset{?}=s\} \overset{FR}\Rightarrow\\
 \{v \overset{?}= h(h(x)),\,v \overset{?}= v_1 + v_3,\, v_1 \overset{?}= v_2+w,\,v_3 \overset{?}= y+z, v_2 \overset{?}=s\}\overset{FR}\Rightarrow\\
 \{v \overset{?}= h(h(x)),\,v \overset{?}= v_1 + v_3,\, v_1 \overset{?}= v_2+v_4,\,v_3 \overset{?}= y+z, v_2 \overset{?}=s, v_4 \overset{?}=w\} \overset{FU}\Rightarrow\\
  \{v \overset{?}= h(v_5),\,v \overset{?}= v_1 + v_3,\, v_1 \overset{?}= v_2+v_4,\,v_3 \overset{?}= y+z, v_2 \overset{?}=s, v_4 \overset{?}=w, v_5 \overset{?}=h(x)\}.$\\
 We see that each equation in the set
 $ \{v \overset{?}= h(v_5),\,v \overset{?}= v_1 + v_3,\, v_1 \overset{?}= v_2+v_4,\,v_3 \overset{?}= y+z,\\ v_2 \overset{?}=s, v_4 \overset{?}=w, v_5 \overset{?}=h(x)\}$
 is in the flattened form.}
 \end{example}
 \par
 \noindent
 \subsubsection{Update h-Depth Set}
 \hfill\par
 We also present a set of inference rules to update the h-depth set. These rules are performed eagerly. \ignore{We apply these rules immediately after applying any other rule in the inference system.}\\ \\
 \fbox{
\begin{minipage}{\textwidth}
\smallskip
\textbf{Update $h$ (U$h$)}
$$ \frac{ \{x \overset{?}= h(y)\}\,\cup \,\Gamma || \{(x,\,c_1),\,(y,\,c_2)\}\cup \triangle || \sigma}{ \{ x \overset{?}= h(y)\}\,\cup \,\Gamma || \{(x,\,c_1 ),\, (y,\,c_1 + 1)\}\cup \triangle || \sigma} 
\text{\hspace{35pt} If $ c_2 < (c_1 + 1)$}$$ 
\smallskip
\end{minipage}
}\\
\begin{example}
\emph{Solve the unification problem: $\{x\overset{?}=h(h(h(y)))\}$.\\ \\
We only consider the pair $\Gamma||\triangle$ since $\sigma$ does not change at this step.\\ \\
 $\{x\overset{?}=h(h(h(y)))\}|| \{(x,\,0),\,(y,\,0)\} \overset{FU^{+}}\Rightarrow \\
 \{x\overset{?}=h(v),\,v\overset{?}= h(v_1),\, v_1\overset{?}=h(y) \}|| \{(x,\,0),\,(y,\,0),\,(v,\,0),\,(v_1,\,0)\} \overset{Uh}\Rightarrow\\
\{x\overset{?}=h(v),\,v\overset{?}= h(v_1),\, v_1\overset{?}=h(y) \}|| \{(x,\,0),\,(y,\,0),\,(v,\,1),\,(v_1,\,0)\}\overset{Uh}\Rightarrow\\
\{x\overset{?}=h(v),\,v\overset{?}= h(v_1),\, v_1\overset{?}=h(y) \}|| \{(x,\,0),\,(y,\,0),\,(v,\,1),\,(v_1,\,1)\}\overset{Uh}\Rightarrow\\
\{x\overset{?}=h(v),\,v\overset{?}= h(v_1),\, v_1\overset{?}=h(y) \}|| \{(x,\,0),\,(y,\,2),\,(v,\,1),\,(v_1,\,1)\}\overset{Uh}\Rightarrow\\
\{x\overset{?}=h(v),\,v\overset{?}= h(v_1),\, v_1\overset{?}=h(y) \}|| \{(x,\,0),\,(y,\,2),\,(v,\,1),\,(v_1,\,2)\}\overset{Uh}\Rightarrow\\
\{x\overset{?}=h(v),\,v\overset{?}= h(v_1),\, v_1\overset{?}=h(y) \}|| \{(x,\,0),\,(y,\,3),\,(v,\,1),\,(v_1,\,2)\},$\\
where $\overset{FU^{+}}\Rightarrow$ represents the application of $FU$ rule once or more than once.\\
\indent It is true that the h-Depth of $y$ is 3 since there are three edges labeled $h$ from $x$ to $y$, in the graph $\mathbb{G}(\Gamma)$.}
\end{example}
\smallskip
\noindent
 \fbox{
\begin{minipage}{\textwidth}
\smallskip
\textbf{Update $+$}
\begin{enumerate}
\item \textbf{Update Left $+$ (UL)} $$\frac{ \{x_1 \overset{?}= y_1 + y_2\}\,\cup \,\Gamma || \{(x_1,\,c_1),\, (y_1,\,c_2),\, (y_2,\,c_3)\} \cup \triangle || \sigma}{ \{x_1 \overset{?}= y_1 + y_2\}\,\cup \,\Gamma || \{(x_1,\, c_1),\,(y_1,\,c_1),\, (y_2,\,c_3)\} \cup \triangle || \sigma} 
 \text{\hspace{5pt} If $ c_2 < c_1 $} $$
\item \textbf{Update Right $+$ (UR)} $$\frac{ \{x_1 \overset{?}= y_1 + y_2\}\,\cup \,\Gamma || \{(x_1,\,c_1),\, (y_1,\,c_2),\, (y_2,\,c_3)\} \cup \triangle || \sigma}{ \{x_1 \overset{?}= y_1 + y_2\}\,\cup \,\Gamma || \{(x_1,\, c_1),\, (y_1,\,c_2),\,(y_2,\,c_1)\} \cup \triangle || \sigma} 
 \text{\hspace{5pt} If $ c_3 < c_1 $}$$
  \end{enumerate}
  \smallskip
 \end{minipage}
}
 \begin{example}
\emph{Solve the unification problem $\{z\overset{?}= x + y,\, x_1 \overset{?}= h(h(z)) \}.$ \\ \\
Similar to the last example, we only consider the pair $\Gamma||\triangle,$\\ \\
 $\{z\overset{?}= x + y,\, x_1 \overset{?}= h(h(z)) \}||\{(x,\,0),(y,\,0),(z,\,0),(x_1,\,0)\}\overset{FU}{\Rightarrow}\\
 \{z\overset{?}= x + y,\, x_1 \overset{?}= h(v),\, v \overset{?}= h(z) \}||\{(x,\,0),(y,\,0),(z,\,0),(x_1,\,0),(v,\,0)\}\overset{Uh^{+}}\Rightarrow\\
 \{z\overset{?}= x + y,\, x_1 \overset{?}= h(v),\, v \overset{?}= h(z) \}||\{(x,\,0),(y,\,0),(z,\,2),(x_1,\,0),(v,\,1)\}\overset{UL}{\Rightarrow}\\
 \{z\overset{?}= x + y,\, x_1 \overset{?}= h(v),\, v \overset{?}= h(z) \}||\{(x,\,2),(y,\,0),(z,\,2),(x_1,\,0),(v,\,1)\}\overset{UR}{\Rightarrow}\\
 \{z\overset{?}= x + y,\, x_1 \overset{?}= h(v),\, v \overset{?}= h(z) \}||\{(x,\,2),(y,\,2),(z,\,2),(x_1,\,0),(v,\,1)\}.$\\ 
 \indent Since there are two edges labeled $h$ from $x_1$ to $z$ in the graph $\mathbb{G}(\Gamma)$, the h-Depth of $z$ is 2. 
The h-Depths of $x$ and $y$ are also updated accordingly.}
\end{example}
Now, we resume the inference procedure for Example~\ref{ex:flatten} and also we consider $\triangle$ because it will be updated at this step.\\ 
$\{v \overset{?}= h(v_3),\, v_3 \overset{?}= h(x),\,v \overset{?}= v_1 + v_2,\, v_1 \overset{?}= s+w,\,v_2 \overset{?}= y+z\}||\\\{(x,\,0),(y,\,0),(z,\,0),(s,\,0),(w,\,0),(v,\,0),(v_1,\,0),(v_2,\,0),(v_3,\,0)\}
\overset{Uh}\Rightarrow\\
\{v \overset{?}= h(v_3),\, v_3 \overset{?}= h(x),\,v \overset{?}= v_1 + v_2,\, v_1 \overset{?}= s+w,\,v_2 \overset{?}= y+z\}||\\\{(x,\,1),(y,\,0),(z,\,0),(s,\,0),(w,\,0),(v,\,0),(v_1,\,0),(v_2,\,0),(v_3,\,0)\}
\overset{Uh}\Rightarrow\\
\{v \overset{?}= h(v_3),\, v_3 \overset{?}= h(x),\,v \overset{?}= v_1 + v_2,\, v_1 \overset{?}= s+w,\,v_2 \overset{?}= y+z\}||\\\{(x,\,1),(y,\,0),(z,\,0),(s,\,0),(w,\,0),(v,\,0),(v_1,\,0),(v_2,\,0),(v_3,\,1)\}
\overset{Uh}\Rightarrow\\
\{v \overset{?}= h(v_3),\, v_3 \overset{?}= h(x),\,v \overset{?}= v_1 + v_2,\, v_1 \overset{?}= s+w,\,v_2 \overset{?}= y+z\}||\\\{(x,\,2),(y,\,0),(z,\,0),(s,\,0),(w,\,0),(v,\,0),(v_1,\,0),(v_2,\,0),(v_3,\,1)\}.$\\
\par
\noindent
\subsubsection{Splitting Rule}
\hfill\par
This rule takes the homomorphism theory into account. In this theory, 
we can not solve equation $h(y)\overset{?}= x_1 + x_2$ unless $y$ can be written
as the sum of two new variables $y=v_1+v_2$, where $v_1$ and $v_2$ are in $\mathcal{NV}$. Without loss of generality we generalize it to $n$ variables $x_1,\ldots,x_n$.\\ \\
 \fbox{
\begin{minipage}{\textwidth}
\smallskip
\textbf{Splitting}
 $$\frac{\{x\overset{?}=h(y),x\overset{?}=x_1+ \cdots + x_n \}\cup \Gamma||\triangle|| \sigma}{\{x\overset{?}=h(y), y \overset{?}= v_1+ \cdots + v_n, x_1 \overset{?}= h(v_1), \ldots, x_n \overset{?}= h(v_n)\} \cup \Gamma||\triangle'|| \sigma}$$
where $n>1$, $x \neq y$ and $x \neq x_{i}$ for any $i$, $\triangle'= \{(v_1,\,0),\ldots,(v_n,\,0)\}\cup\triangle$, and $v_1,\ldots, v_n$ are fresh variables in $\mathcal{NV}$.
 \smallskip
  \end{minipage}
  }\\
 \begin{example} \emph{ Solve the unification problem
 $\{ h(h(x)) \overset{?}= y_1+y_2\}.$}\\ \par
 \emph{Still we only consider pair $\Gamma || \triangle$, since rules modifying $\sigma$ are not introduced yet.}\\ \\
$ \{ h(h(x)) \overset{?}= y_1+y_2 \}||\{(x,\,0),(y_1,\,0),(y_2,\,0)\}\, \overset{FBS^{+}}\Rightarrow$ \\
$\{ v \overset{?}=h(v_{1}), \, v_{1}\overset{?}=h(x),\, v\overset{?}=y_{1}+y_{2}\ \}||\{(x,\,0), (y_{1}, \, 0), (y_{2}, \, 0), (v,\,0), (v_{1}, \,0)\} \, \overset{Uh^{+}}\Rightarrow$\\
 $\{ v \overset{?}=h(v_{1}), \, v_{1}\overset{?}=h(x),\, v\overset{?}=y_{1}+y_{2}\ \}|||\{(x,\,2),(y_1,\,0),(y_2,\,0),(v,\,0),(v_1,\,1)\}, \,\overset{Splitting}\Rightarrow$\\
 $\{ v \overset{?}=h(v_1),\, v_1 \overset{?}= v_{11} + v_{12},\, y_1\overset{?}= h(v_{11}),\, y_2 \overset{?}=h(v_{12}),\, v_1 \overset{?}= h(x) \, \}||\\\{(x,\,2),(y_1,\,0),(y_2,\,0),(v,\,0),(v_1,\,1),(v_{11},\,0),(v_{12},\,0)\}, \,\overset{Uh^{+} }\Rightarrow$\\
$ \{ v \overset{?}=h(v_1),\, v_1 \overset{?}= v_{11} + v_{12},\, y_1\overset{?}= h(v_{11}),\, y_2 \overset{?}=h(v_{12}),\, v_1 \overset{?}= h(x) \, \}||\\\{(x,\,2),(y_1,\,0),(y_2,\,0),(v,\,0),(v_1,\,1),(v_{11},\,1),(v_{12},\,1)\}, \,\overset{Splitting}\Rightarrow$\\
$ \{ v \overset{?}=h(v_1),\, y_1\overset{?}= h(v_{11}),\, y_2 \overset{?}=h(v_{12}),\, v_1 \overset{?}= h(x),\, x \overset{?}= v_{13} + v_{14}, \, v_{11}\overset{?}=h(v_{13}),\\v_{12}\overset{?}=h(v_{14}) \, \}||\{(x,\,2),(y_1,\,0),(y_2,\,0),(v,\,0),(v_1,\,1),(v_{11},\,1),(v_{12},\,1),(v_{13},\,0),(v_{14},\,0)\}\overset{Uh^{+}}\Rightarrow$
 $ \{ v \overset{?}=h(v_1),\, y_1\overset{?}= h(v_{11}),\, y_2 \overset{?}=h(v_{12}),\, v_1 \overset{?}= h(x),\, x \overset{?}= v_{13} + v_{14}, \, v_{11}\overset{?}=h(v_{13}),\\v_{12}\overset{?}=h(v_{14}) \, \}||\{(x,\,2),(y_1,\,0),(y_2,\,0),(v,\,0),(v_1,\,1),(v_{11},\,1),(v_{12},\,1),(v_{13},\,2),(v_{14},\,2)\}$.
\end{example}
\subsubsection{Trivial}
\hfill\par
The Trivial inference rule is to remove trivial equations in the given problem $\Gamma$. 
\\ \\
 \fbox{
\begin{minipage}{\textwidth}
\smallskip
$$\frac{\{ t\overset{?}= t\}\,\cup \,\Gamma || \triangle || \sigma}{\Gamma || \triangle || \sigma} $$ 
 \smallskip
  \end{minipage}
  }\\
\par \noindent
\subsubsection{Variable Elimination (VE)}
\hfill \par
The Variable Elimination rule is to convert the equations into assignments. In other words, it is used to find the most general unifier. \\ \\
 \fbox{
\begin{minipage}{\textwidth}
\smallskip
\begin{enumerate}
\item \textbf{VE1}$$\frac{\{ x \overset{?}= y \}\,\cup \, \Gamma || \triangle || \sigma}{\Gamma \{x \mapsto y\} || \triangle || \sigma \{ x \mapsto y \} \cup \{x \mapsto y \} } \text{\hspace{25pt} if $x$ and $y$ are distinct variables} $$
\item \textbf{VE2} $$\frac{\{ x \overset{?}= t \}\,\cup \,\Gamma || \triangle || \sigma}{\Gamma \{ x \mapsto t\} || \triangle || \sigma \{ x \mapsto t\}\cup\{x \mapsto t\} } \text{\hspace{25pt} if $t \notin \mathcal{V}$ and $x$ does not occur in $t$}$$
\end{enumerate}
\smallskip
  \end{minipage}
  }\\ \par
The rule VE2 is performed last after all other inference rules have been performed. The rule VE1 is performed eagerly.
\begin{example}
\emph{ Solve unification problem $\{ x \overset{?}= y ,\, x \overset{?}= h(z) \}$.} \\
\par\noindent
 $\{ x \overset{?}= y ,\, x \overset{?}= h(z) \}||\{(x,\,0),\,(y,\,0),\,(z,\,0)\}||\emptyset \overset{Uh}\Rightarrow \\
 \{ x \overset{?}= y ,\, x \overset{?}= h(z) \}||\{(x,\,0),\,(y,\,0),\,(z,\,1)\}||\emptyset \overset{VE1}\Rightarrow \\
 \{y \overset{?}= h(z)\}||\{(x,\,0),\,(y,\,0),\,(z,\,1)\}||\{ x \mapsto y\}\overset{VE2}\Rightarrow\\
 \emptyset||\{(x,\,0),\,(y,\,0),\,(z,\,1)\}||\{ x \mapsto h(z),\,y \mapsto h(z)\}.$\par
\emph{The substitution $\{ x \mapsto h(z),\,y \mapsto h(z)\}$ is the most general unifier of the given problem $\{ x \overset{?}= y ,\,x \overset{?}= h(z) \}$.}
 \end{example}
 \par
\noindent
\subsubsection{Decomposition (Decomp)}
\hfill\par
 The Decomposition rule decomposes an equation into several sub-equations if both sides' top symbol matches.\\ \\
  \fbox{
\begin{minipage}{\textwidth}
\smallskip
\textbf{Decomp}
$$\frac{\{ x \overset{?}= f( s_1,\ldots,s_n), x \overset{?}= f( t_1,\ldots,t_n)\}\cup \Gamma || \triangle || \sigma}{\{x \overset{?}= f( t_1, \ldots,t_n),\, s_1 \overset{?}= t_1,\ldots,s_n \overset{?}= t_n \}
\, \cup \,\Gamma|| \triangle || \sigma} \text{\hspace{25pt} if $f \neq +$ }$$
\smallskip
  \end{minipage}
  }\\
\begin{example}
\emph{Solve the unification problem $\{h(h(x)) \overset{?}= h(h(y))\}$.}\\ \\
$\{h(h(x)) \overset{?}= h(h(y))\}||\{(x,\,0),\,(y,\,0)\}||\emptyset \overset{Flatten^{+}}\Rightarrow \\
\{v \overset{?}= h(v_1),\,v_1 \overset{?}= h(x),\,v \overset{?}= h(v_2),\,v_2 \overset{?}= h(y)\}||\{(x,\,0),\,(y,\,0),(v,\,0),(v_1,\,0),(v_2,\,0)\}||\emptyset\overset{Uh^{+}}\Rightarrow
\{v \overset{?}= h(v_1),\,v_1 \overset{?}= h(x),\,v \overset{?}= h(v_2),\,v_2 \overset{?}= h(y)\}||\{(x,\,2),\,(y,\,2),(v,\,0),(v_1,\,1),(v_2,\,1)\}||\emptyset\overset{Decomp}\Rightarrow
\{v \overset{?}= h(v_1),\, v_1 \overset{?}= v_2,\, v_1 \overset{?}= h(x),\,v_2 \overset{?}= h(y)\}||\{(x,\,2),\,(y,\,2),(v,\,0),(v_1,\,1),(v_2,\,1)\}||\emptyset \overset{VE1}\Rightarrow
\{v \overset{?}= h(v_2),\, v_2 \overset{?}= h(x),\,v_2 \overset{?}= h(y)\}||\{(x,\,2),\,(y,\,2),(v,\,0),(v_1,\,1),(v_2,\,1)\}||\{v_1 \mapsto v_2\}\overset{Decomp}\Rightarrow
\{v \overset{?}= h(v_2),\, v_2 \overset{?}= h(x),\,x\overset{?}= y\}||\{(x,\,2),\,(y,\,2),(v,\,0),(v_1,\,1),(v_2,\,1)\}||\{v_1 \mapsto v_2\}\overset{VE2^{+}}\Rightarrow\\
\emptyset||\{(x,\,2),\,(y,\,2),(v,\,0),(v_1,\,1),(v_2,\,1)\}||\{v_1 \mapsto h(y),\,x\mapsto y,\,v \mapsto h(h(y)),\,v_2 \mapsto h(y)\},$\\
\emph{where $\{x\mapsto y\}$ is the most general unifier of the problem $\{h(h(x)) \overset{?}= h(h(y))\}$.}
\end{example}
\par
\noindent
 \subsubsection{AC Unification}
 \hfill \par
 The AC Unification rule calls an AC unification algorithm to unify the AC part of the problem. Notice that we apply AC unification only once when no other rule except VE-2 can apply.
 In this inference rule $\Psi$ represents the set of all equations with the $+$ symbol on the right hand side.
 $\Gamma$ represents the set of equations not containing a $+$ symbol. \textit{Unify} is a function that returns
 one of the complete set of unifiers returned by the AC unification algorithm. \textit{GetEqs} is a function that
 takes a substitution and returns the equational form of that substitution. In other words, 
 $GetEqs(\{x_1 \mapsto t_1, \ldots, x_n \mapsto t_n\}) = \{x_1 \overset{?}= t_1, \ldots, x_n \overset{?}= t_n\}$. \\ \\
   \fbox{
\begin{minipage}{\textwidth}
\smallskip
\textbf{AC Unification} $$\frac{\Psi \cup \Gamma || \triangle || \sigma}{GetEqs(\theta_{1}) \cup \Gamma|| \triangle || \sigma \vee\ldots \vee GetEqs(\theta_{n}) \cup \Gamma|| \triangle || \sigma }$$
\text{where $\textit{Unify}(\Psi) = \{\theta_{1}, \ldots, \theta_{n}\}.$}
\smallskip
  \end{minipage}
  }\smallskip\\ \\
  We illustrate the applicability of the AC unification rule using the example below. For convenience, we only consider $\Gamma$ from the problem.
  \begin{example} \emph{Solve the unification problem
  $\{x + y \overset{?}= z + y_{1}, \, x_{1} \overset{?}= x_{2} \}$, where $x, y, z, x_{1}, x_{2}$, and $y_{1}$ are pairwise distinct.\\ \\
  $\{x + y \overset{?}= z + y_{1}, \, x_{1} \overset{?}= x_{2} \} \,\overset{FBS}\Rightarrow$ 
  $\{ v \overset{?}= x + y, \, v \overset{?}= z + y_{1}\} \cup\{ x_{1} \overset{?}= x_{2} \}\,\overset{AC\,Unification}\Rightarrow$\\
  $\{ v \overset{?}= c_1 + c_2 + c_3+c_{4} , \, x \overset{?}= c_{1} + c_{2},\, y \overset{?}= c_{3} + c_{4}, z \overset{?}= c_{1} + c_{3},\, y_{1} \overset{?}= c_{2} + c_{4} \} \cup\{ x_{1} \overset{?}= x_{2} \}\, \vee$\\
  $\{ v \overset{?}= c + z + y,\, x \overset{?}= c + z, \, y_{1} \overset{?}= c + y \} \cup\{ x_{1} \overset{?}= x_{2} \}\, \vee$\\
    $\{ v \overset{?}= z + c + y,\, x \overset{?}= z+c, \, y_{1} \overset{?}= c + y \} \cup\{ x_{1} \overset{?}= x_{2} \}\, \vee$\\
    $\{ v \overset{?}= x + c + z,\, y \overset{?}= c+z, \, y_{1} \overset{?}= x+c \} \cup\{ x_{1} \overset{?}= x_{2} \}\, \vee$\\
     $\{ v \overset{?}= x + z + c,\, y \overset{?}= z+c, \, y_{1} \overset{?}= x+c \} \cup\{ x_{1} \overset{?}= x_{2} \}\, \vee$\\
     $\{ v \overset{?}= z + y_{1},\, x \overset{?}= z, \, y \overset{?}= y_{1} \} \cup\{ x_{1} \overset{?}= x_{2} \}\, \vee$\\
     $\{ v \overset{?}= y_{1} + z,\, x \overset{?}= y_{1}, \, y \overset{?}= z \} \cup\{ x_{1} \overset{?}= x_{2} \},$\\
     where $c, c_{1}, c_{2}, c_{3}$, and $c_{4}$ are constant symbols.}
  \end{example}
\ignore{Note that we have written the rule for one member of the complete set of AC unifiers of $\Psi$.
This will branch on every member of the complete set of AC unifiers of $\Psi$.}
\par
\noindent
\subsubsection{Occur Check (OC)}
\hfill \par
OC checks  if a variable on the left-hand side of an equation occurs on the other side of the equation. If it does, then the problem has no solution. This rule has the highest priority.\\ \\
   \fbox{
\begin{minipage}{\textwidth}
\smallskip
\textbf{OC}
$$\frac{\{ x \overset{?}= f( t_1,\ldots,t_n)\}\cup \Gamma || \triangle || \sigma}{ \bot}
\text{\hspace{35pt} If $x \in \mathcal{V}ar(f( t_1,\ldots,t_n)\sigma$)}$$
where $\mathcal{V}ar(f( t_1,\ldots,t_n)\sigma)$ represents set of all variables that occur in $f( t_1,\ldots,t_n)\sigma$.
\smallskip
  \end{minipage}
  }
\begin{example} \emph{Solve the following unification problem
$\{x \overset{?}= y, \, y \overset{?}= z + x \}.$\\ \\
$\{x \overset{?}= y, \, y \overset{?}= z+x\}||\{(x,\,0),(y,\,0),(z,\,0)\}||\emptyset \,\overset{VE1}\Rightarrow\\
\{ y \overset{?}= z + y\}||\{(x,\,0),(y,\,0),(z,\,0)\}||\{x \mapsto y \}\,\overset{OC}\Rightarrow\, \textit{Fail}.$\\
Hence, the problem $\{x \overset{?}= y, \, y \overset{?}= z + x\}$ has no solution.}
\end{example}
\par
\noindent
\subsubsection{Clash}
\hfill \par
This rule checks if the \textit{top symbol} on both sides of an equation is the same. If not, then there is no solution to the problem, unless one of them is $h$ and the other $+$. \\ \\
 \fbox{
\begin{minipage}{\textwidth}
\smallskip
\textbf{Clash}
$$\frac{\{ x \overset{?}= f( s_1,\ldots,s_m), \, x \overset{?}= g( t_1,\ldots,t_n)\}\cup \Gamma || \triangle || \sigma}{\bot} 
\text{\hspace{25pt} If $f \notin \{h,\,+\}$ or $g \notin \{h,\,+\}$ } $$ 
\smallskip
  \end{minipage}
  }
\begin{example} \emph{Solve the unification problem
$\{ f(x,\, y) \overset{?}= g(h(z)) \}$, where $f$ and $g$ are two distinct uninterpreted function symbols.\\ \\
$\{ f(x,\, y) \overset{?}= g(h(z)) \}||\{(x,\,0),(y,\,0),(z,\,0) \} ||\emptyset \overset{Flatten^{+}}\Rightarrow\\
\{ v \overset{?}= f(x,\, y),\,v \overset{?}= g(v_1), v_1 \overset{?}=h(z) \} ||\{(x,\,0),(y,\,0),(z,\,0),(v,\,0),(v_1,\,0) \} ||\emptyset \overset{Uh^{+}}\Rightarrow\\
\{ v \overset{?}= f(x,\, y),\,v \overset{?}= h(v_1),\, v_1 \overset{?}=h(z) \} ||\{(x,\,0),(y,\,0),(z,\,1),(v,\,0),(v_1,\,0) \} ||\emptyset \overset{Clash}\Rightarrow \textit{Fail}.$
Hence, the problem $\{ f(x,\, y) \overset{?}= g(h(z)) \}$ has no solution.}
\end{example}
\par
\noindent
\subsubsection{Bound Check (BC)}
\hfill\par
The Bound Check is to determine if a solution exists within the bound $\kappa$, a given maximum h-depth of any variable in $\Gamma$.
If one of the h-depths in the h-depth set $\triangle$ exceeds the bound $\kappa$, then the problem has no solution. We apply this rule immediately after the rules of update h-depth set.\\ \\
 \fbox{
\begin{minipage}{\textwidth}
\smallskip
\textbf{BC}
$$ \frac{ \Gamma || \triangle || \sigma}{\bot} 
\text{\hspace{35pt} If $MaxVal(\triangle) > \kappa$}$$
\smallskip
  \end{minipage}
  }\\
\begin{example}
 \emph{Solve the following unification problem $\{ h(y) \overset{?}= y + x \}$.\\ \\
 Let the bound be $\kappa =2$.\\
 $\{ h(y) \overset{?}= y + x \}||\{(x,\,0),(y,\,0)\}||\emptyset\overset{FBS}\Rightarrow\\
 \{ v\overset{?}= h(y),\,v \overset{?}= y + x\}||\{(x,\,0),(y,\,0),(v,\,0)\}||\emptyset\overset{Uh}\Rightarrow\\
 \{ v\overset{?}= h(y),\,v \overset{?}= y + x\}||\{(x,\,0),(y,\,1),(v,\,0)\}||\emptyset\overset{Splitting}\Rightarrow\\
 \{ v\overset{?}= h(y),\, y \overset{?}= v_{11} + v_{12},\, y \overset{?}= h(v_{11}),\, x \overset{?}= h(v_{12})||\{(x,\,0),(y,\,1),(v,\,0),(v_{11},\,0),(v_{12},\,0)\}||\emptyset \overset{Uh^{+}}\Rightarrow
  \{ v\overset{?}= h(y),\, y \overset{?}= v_{11} + v_{12},\, y \overset{?}= h(v_{11}),\, x \overset{?}= h(v_{12})||\{(x,\,0),(y,\,1),(v,\,0),(v_{11},\,2),(v_{12},\,1)\}||\emptyset \overset{\scriptsize{Splitting}}\Rightarrow
 \{ v\overset{?}= h(y),\, v_{11} \overset{?}= v_{13} + v_{14},\, v_{11} \overset{?}= h(v_{13}),\, v_{12}\overset{?}= h(v_{14}),\, y \overset{?}= h(v_{11}),\, x \overset{?}= h(v_{12})||\\\{(x,\,0),(y,\,1),(v,\,0),(v_{11},\,2),(v_{12},\,1),(v_{13},\,0),(v_{14},\,0)\}||\emptyset \overset{Uh^{+}}\Rightarrow\\
  \{ v\overset{?}= h(y),\, v_{11} \overset{?}= v_{13} + v_{14},\, v_{11} \overset{?}= h(v_{13}),\, v_{12}\overset{?}= h(v_{14}),\, y \overset{?}= h(v_{11}),\, x \overset{?}= h(v_{12})||\\\{(x,\,0),(y,\,1),(v,\,0),(v_{11},\,2),(v_{12},\,1),
 (v_{13},\,3),(v_{14},\,2)\}||\emptyset \overset{BC}\Rightarrow \textit{Fail}.$\\
  Since $MaxVal(\triangle) = 3 > \kappa$, the problem $\{ h(y) \overset{?}= y + x\}$ has no solution within the given bound.}
\end{example}
\par
\noindent
\subsubsection{Orient}
\hfill\par
The Orient rule swaps the left side term of an equation with the right side term. In particular, when the left side term is a variable but not the right side term.\\ \\
 \fbox{
\begin{minipage}{\textwidth}
\smallskip
\textbf{Orient}
$$ \frac{ \{t \overset{?}{=} x\} \cup\Gamma || \triangle || \sigma}{\{x \overset{?}{=} t\} \cup\Gamma || \triangle || \sigma} 
\text{\hspace{35pt} If $t$ is not a variable}$$
\smallskip
  \end{minipage}
  }\\

\begin{algorithm}
\caption{AChUnify}
\paragraph{\textbf{Input:}}
\begin{itemize}
\item An equation set $\Gamma$, a bound $\kappa$, an empty set $\sigma$, and an empty h-depth set $\triangle$.
\end{itemize}
\paragraph{\textbf{Output:}}
\begin{itemize}
\item A complete set of $\kappa$-bounded $ACh$ unifiers $\{\sigma_1,\ldots, \sigma_n\}$ or $\bot$ indicating that the problem has no solution.
\end{itemize}
1: Apply $Trivial$ to eliminate equations of the form $t \overset{?}{=} t$.\\
2: Apply $OC$ to see if any variable on the left side occurs on the right. If yes, \\ \hspace*{0.4cm}then return $\bot$.\\
3: Flatten the set of equations $\Gamma$ using the flattening rules.\\
4: Update the h-depth set $\triangle$.\\
5: Apply $BC$ to see if $MaxVal(\triangle) > \kappa$. If yes, then return $\bot$.\\
6: Apply the $Orient$ rule.\\
7: Apply the $Splitting$ rule.\\
8: Apply the $Clash$ rule.\\
9: Apply the $Decomposition$ rule.\\
10: Apply the $AC\: Unification$ rule.\\
11: Finally, apply the $Variable\,Elimination$ rule and get the output.\\
\end{algorithm}
\section{Proof of Correctness}
\ignore{\begin{tcolorbox}[colback=yellow!100]
To check, if these proofs comply with the reviewer's comments!
 \end{tcolorbox}}
We prove that the proposed inference system is terminating, sound, and complete.
\subsection{Termination}
Before going to present the proof of termination, we shall introduce few notation which will be used in the subsequent sections. For two set triples, $ \Gamma||\triangle|| \sigma$ and $\Gamma'||\triangle'|| \sigma'$,
\begin{itemize}
\item $\Gamma||\triangle|| \sigma \Rightarrow_{\Iach} \Gamma'||\triangle'|| \sigma'$, means that the set triple $\Gamma'||\triangle'|| \sigma'$ is deduced from $\Gamma||\triangle|| \sigma$ by applying a rule from $\Iach$ once. We call it as one step.
\item $\Gamma||\triangle|| \sigma \overset{*}{\Rightarrow_{\Iach}} \Gamma'||\triangle'|| \sigma'$, means that the set triple $\Gamma'||\triangle'|| \sigma'$ is deduced from $\Gamma||\triangle|| \sigma$ by zero or more steps
\item $\Gamma||\triangle|| \sigma \overset{+}{\Rightarrow_{\Iach}} \Gamma'||\triangle'|| \sigma'$, means that the set triple $\Gamma'||\triangle'|| \sigma'$ is deduced from $\Gamma||\triangle|| \sigma$ by one or more steps
\end{itemize}
As we notice, AC unification divides $\Gamma||\triangle|| \sigma$ into finite number of branches $\Gamma_{1}||\triangle_{1}|| \sigma_{1}$ and so on $\Gamma_{n}||\triangle_{n}|| \sigma_{n}$. Hence, for a triple $\Gamma||\triangle|| \sigma$, after applying some inference rules, the result is a disjunction of triples $\bigvee_{i}(\Gamma_{i}||\triangle_{i}|| \sigma_{i})$. Accordingly, we introduce the following notation:
\begin{itemize}
\item $\striple {\scriptstyle \implies_{\Iach}} \bigvee_{i}(\Gamma_{i}||\triangle_{i}|| \sigma_{i})$, where $\bigvee_{i}(\Gamma_{i}||\triangle_{i}|| \sigma_{i})$ is a disjunction of triples, means that the set triple $\striple$ becomes $\bigvee_{i}(\Gamma_{i}||\triangle_{i}|| \sigma_{i})$ with an application of a rule once.
\item $\striple \overset{+} {\scriptstyle \implies_{\Iach}} \bigvee_{i}(\Gamma_{i}||\triangle_{i}|| \sigma_{i})$ means that $\striple$ becomes $\bigvee_{i}(\Gamma_{i}||\triangle_{i}|| \sigma_{i})$ after applying some inference rules once or more than once.
\item $\striple \overset{*} {\scriptstyle \implies_{\Iach}} \bigvee_{i}(\Gamma_{i}||\triangle_{i}|| \sigma_{i})$ means that $\striple$ becomes $\bigvee_{i}(\Gamma_{i}||\triangle_{i}|| \sigma_{i})$ after applying some inference rules zero or more times.
\end{itemize}
\par
Here, we define a measure of $\striple$ for proving termination:
\begin{itemize}
\ignore{\item Let $|\Gamma|$ be the cardinality of $\Gamma$. Since $|\Gamma|$ is a natural number, $|\Gamma|$ with $\leq$ on natural numbers is a well-founded ordering.}
\item Let $Sym(\Gamma)$ be a multi-set of non-variable symbols occurring in $\Gamma$. The standard ordering of $|Sym(\Gamma)|$ based on natural numbers is a well-founded ordering on the set of equations.
\ignore{\item Let $Top(\Gamma)$ be the set of all top symbols of $\Gamma$. Since $|Top(\Gamma)|$ is a natural number, $|Top(\Gamma)|$ with $\leq$ on natural numbers is a well-founded ordering.}
\item Let $\kappa$ be a natural number. Let $\overline{h_{d}}(\Gamma):= \{ (\kappa+1) \,\text{-} \,h_{d}(x, \Gamma)\mid(x, h_{d}(x, \Gamma))\in h_{d}(\Gamma)\}$ be a multi-set. Since every element of the set is a natural number, the multi-set order for $\overline{h_{d}}(\Gamma)$ is a well-founded ordering.
\item Let $p$ be a number of non-solved variables in $\Gamma$. 
\item Let $m$ be the number of equations of the form $f(t)\overset{?}=x$ in $\Gamma$.
\item Let $n$ be the number of +-equations with $x$ occurring on the left side, i.e, $x = x_1+\cdots+x_n$.
\end{itemize}
\par
Then we define the measure of $\striple$ as the following:
$$\mstriple = (n, |Sym(\Gamma)|, p, m, |\Gamma|, \overline{h_{d}}(\Gamma)).$$
Since each element in this tuple with its corresponding order is well-founded, the lexicographic order on this tuple is well-founded as well.
\ignore{As we apply the rules of flattening whenever needed, we assume that the equations of $\Gamma$ are in flattened form.}

We show that $\mstriple$ decreases with the application of each rule of the inference system $\Iach$ except AC unification. The reader can see the proof of termination of the AC unification in~\cite{FF}.

\begin{lemma}
\label{lem:ter1}
\emph{Let $\striple$ and $\striplep$ be two set triples, where $\Gamma$ and $\Gamma'$ are in flattened form, such that $\striple \Rightarrow_{\Iach} \striplep$. Then $\mstriple > \mstriplep$.}
\end{lemma}
\begin{proof}
\textbf{Trivial.} The cardinality of $\Gamma$, $|\Gamma |$, decreases while other components of the measure either stays the same or decreases.
Hence, $\mstriple > \mstriplep$.\\
\textbf{Decomposition.} The number of $f$ symbols decreased by one, and hence $|Sym(\Gamma)|$ decreases while $p$ stays the same. Hence, $\mstriple > \mstriplep$.\\
\textbf{Update h-Depth Set.} 
On application of one of the update rules, increases h-depth of a variable $x$ from $n$ to $n+1$. However, $\kappa \text{-} n >\kappa \text{-} (n+1)$. Which means that $\overline{h_{d}}(\Gamma)$ decreases while the other components stay the same. Hence, $\mstriple > \mstriplep$.\\
\noindent
\textbf{Splitting.} On the application of the Splitting rule, $n$, the number of +-equations with $x$ on the left side decreased by one. So, $\mstriple > \mstriplep$.\\
\textbf{Orient.} It is not difficult to see the fact that $m$ decreases. \\
\textbf{Variable elimination.} Of course, the number of non-solved variables decreases in the application of this rule.\\
\end{proof}
\begin{theorem}[\textbf{Termination}]
\emph{For any set triple $\striple$, there is a set triple $\striplep$ such that $\Gamma||\triangle|| \sigma \overset{*}{\Rightarrow_{\Iach}} \Gamma'||\triangle'|| \sigma'$ and none of the rules $\Iach$ can be applied on $\striplep$.}
\end{theorem}
\begin{proof}
By induction on Lemma~\ref{lem:ter1}, this theorem can be proved.

\ignore{First notice that at some point all the Decomposition rules not involving $h$ will eventually be performed.
That is because when we perform Decomposition on the top symbol $f$, one occurrence of $f$ disappears, and none of the rules can make them come back. So from now on, we assume that all such rules have been performed.

Let us call $x \overset{?}= h(y)$ an {\em $h$-rule of depth $i$} if $h_d(x) = i$. A splitting rule involving $x \overset{?}= h(y)$ and $x \overset{?}= v_1 +\cdots+ v_n$ 
is called a {\em splitting rule of depth $i$}. We say that the variable $x$ is split.
A Decomposition rule involving $x \overset{?}= h(y)$ and $x \overset{?}= h(z)$ is called an
{\em $h$-Decomp rule of depth $i$}. We say that $x$ is $h$-decomposed.  

We first show that at some point all splitting rules of depth 0 and $h$-Decomp rules of depth 0 will have been performed. We notice that after AC Unification is called for the first time, any variable appearing in a VarVar equation or a $+$-equation will either appear exactly once in the left-hand side of one of those equations and never on the right-hand side, or else it will never appear on the left-hand side of one of those equations. This is because the AC unification rule creates equations from substitutions, which have this property. Also, the AC unification rule, the VE1 rule, and the Trivial rule will not change this property. Therefore when a variable $x$ is split, all of the occurrences of $x$ in $+$ equations and VarVar equations disappear. If $x$ has depth 0, then it cannot occur on the right-hand side of an $h$-equation. So after this split, variable $x$ cannot be split anymore. Also, $h$ can then only be $h$-decomposed a finite number of times, because each time eliminates an $h$-equation with $x$ on the left-hand side, and no new ones of depth 0 can be created.

We want to show by induction that all splits and $h$-Decomps will eventually be performed. Suppose that all of them at depth $i$ has been performed at some point. We will show that at some point all of them at depth $i+1$ will be performed.
 Again we notice that after AC Unification is called for the first time, any variable appearing in a VarVar equation or a $+$-equation will either appear exactly once in the left-hand side of one of those equations and never on the right-hand side, or else it will never appear on the left-hand side of one of those equations. This is because the AC unification rule creates equations from substitutions, which have this property. Also, the AC unification rule, the VE1 rule, and the Trivial rule will not change this property. Therefore when a variable $x$ is split, all of the occurrences of $x$ in $+$ equations and VarVar equations disappear. If $x$ has depth $i+1$, then it cannot occur on the right-hand side of an $h$-equation that can possibly be split again. This is a result of our induction hypothesis. So after this split, variable $x$ cannot be split anymore. Also, $h$ can then only be $h$-decomposed a finite number of times, because each time eliminates an $h$-equation with $x$ on the left-hand side, and no new ones of depth $i+1$ can be created.
 
 Because of our Bound Check rule, all splits and $h$-Decomp rules will eventually be performed. From now on, we assume that they have all been performed. Assume all VE1 rules and Trivial rules that currently exist been performed. Then suppose we perform AC Unification. This will not create any applications of Trivial or VE1. Therefore the process will be finished here. The only thing left is the performance of VE2 rules at the end, which trivially halts because they reduce the number of equations. }
\end{proof}
\subsection{Soundness}
\noindent In this Section, we show that our inference system $\Iach$ is truth-preserving.
\begin{lemma}\label{lem:soundness1}
\emph{Let $\striple$ and $\striplep$ be two set triples such that $\striple \Rightarrow_{\Iach} \striplep$ via all the rules of $\Iach$ except AC unification. Let $\theta$ be a substitution such that $\theta \models \striplep$. Then $\theta \models \Gamma ||\triangle||\sigma$.}
\end{lemma}
\begin{proof}
\textbf{Trivial.} It is trivially true.\\
\ignore{Splitting.

 $$\frac{\{w\overset{?}=h(y),w\overset{?}=x_1+ \cdots + x_n \}\cup \Gamma||\triangle|| \sigma}{\{w\overset{?}=h(y), y \overset{?}= v_1+ \cdots + v_n, x_1 \overset{?}= h(v_1), \ldots, x_n \overset{?}= h(v_n)\} \cup \Gamma||\triangle'|| \sigma}$$
where $n>1$, $y \neq w$, $\triangle'= \{(v_1,\,0),\ldots,(v_n,\,0)\}\cup\triangle\}$, and $v_1,\ldots, v_n$ are fresh variables in $\mathcal{NV}$.

\ignore{$$\frac{\{w\overset{?}=h(y),w\overset{?}=x_1+x_2 \}\cup \Gamma||\triangle|| \sigma}{\{w\overset{?}=h(y), y \overset{?}= v_1+v_2, x_1 \overset{?}= h(v_1), x_2 \overset{?}= h(v_2)\} \cup \Gamma||\triangle'|| \sigma}$$}}
\textbf{Splitting.} Let $\theta$ be a substitution. Assume that $\theta$ satisfies $\{w\overset{?}=h(y), y \overset{?}= v_1+ \cdots + v_n, x_1 \overset{?}= h(v_1), \ldots, x_n \overset{?}= h(v_n)\} \cup \Gamma$.
Then we have that $w\theta\overset{?}=h(y)\theta$, $y\theta\overset{?}= (v_1+\cdots +v_n)\theta$, $x_1\theta \overset{?}= h(v_1)\theta$, \ldots, $x_n \theta \overset{?}= h(v_n)\theta$.
This implies that $w\theta\overset{?}=h(y\theta)$, $y\theta\overset{?}= v_1\theta+\cdots+v_n\theta$, $x_1\theta \overset{?}= h(v_1\theta)$ \ldots $x_n \theta \overset{?}= h(v_n\theta)$.
In order to prove that $\theta$ satisfies $\{w\overset{?}=h(y),w\overset{?}=x_1+\cdots+x_n \}$, it is enough to prove $\theta$
satisfies the equation $w\overset{?}=x_1+\cdots+x_n$. By considering the right side term $x_1+\cdots+x_n$ and after applying the substitution, we get $(x_1+\cdots+x_n)\theta \overset{?}= x_1\theta +\cdots+ x_n \theta \overset{?}= h(v_1\theta) +\cdots+ h(v_n\theta)$.
By the homomorphism theory, we write that $h(v_1\theta) +\cdots+ h(v_n\theta) \overset{?}= h(v_1\theta +\cdots+ v_n \theta)$.
Then $h(v_1\theta +\cdots+ v_n \theta)\overset{?}= h(y\theta)\overset{?}= w\theta$.
Hence, $\theta$ satisfies $w\overset{?}=x_1+\cdots+x_n$.\\
\textbf{Variable Elimination.}\\
{\textbf{VE1.}} Assume that $\theta \models \Gamma\{x \mapsto y\} || \triangle || \sigma\{ x \mapsto y\} \cup \{x \mapsto y \}$. This means that
$\theta$ satisfies $\Gamma\{x \mapsto y\}$ and $\sigma \{ x \mapsto y\} \cup \{x \mapsto y \}$.
Now, we have to prove that $\theta$ satisfies $\{ x \overset{?}= y \},\Gamma,$ and $\sigma$.
But $\theta$ satisfies $x \mapsto y$ means that $x \theta \overset{?}= y\theta$.
$\Gamma$ is $\Gamma \{x \mapsto y\}$ but without replacing $x$ with $y$.
Since $y \theta \overset{?}= x\theta$, the substitution $\theta$ satisfies $y \mapsto x$.
Hence, we conclude that $\theta$ satisfies $\Gamma$ and $\sigma$.\\
\textbf{VE2.} We have that $\theta$ satisfies $\Gamma$ and $\sigma\{ x \mapsto t\} \cup \{x \mapsto t \}$.
Now, we have to prove that $\theta$ satisfies $\{ x \overset{?}= t \}$ and $\sigma$.
By the definition of $\theta \models \Gamma$, we have $x \theta \overset{?}= t\theta$ and it is enough to prove that
$\theta$ satisfies $\sigma$.
Let $w \mapsto s[x]$ be an assignment in $\sigma$. 
After applying $x \mapsto t$ on $\sigma$, the assignment $y \mapsto s$ with $s|_p=x$, where $p$ is a position, becomes $y \mapsto s[t]_{p}$.
We also know that $\theta$ satisfies $\sigma\{ x \mapsto t\}$ implies that $\theta$ also satisfies $w \mapsto s[t]_{p}$.
Then by the definition, we write that $y \theta \overset{?}= s[t\theta]_{p} \overset{?}= s[x\theta]_{p}$. This means that $\theta$
satisfies the assignment $w \mapsto s[x]$. Hence, $\theta$ satisfies $\sigma$.\\
\noindent
\textbf{Decomposition.}
Assume that $\theta \models \{x \overset{?}= f( t_1, \ldots,t_n),\, s_1 \overset{?}= t_1,\ldots,s_n \overset{?}= t_n \} \, \cup \,\Gamma|| \triangle || \sigma.$
This means that $\theta$ satisfies $\{x \overset{?}= f( t_1, \ldots,t_n),\, s_1 \overset{?}= t_1,\ldots,s_n \overset{?}= t_n \} \, \cup \,\Gamma$.
Now we have to prove that $\theta$ satisfies $\{ x \overset{?}= f( s_1, \ldots,s_n), x \overset{?}= f( t_1, \ldots,t_n)\}\cup \Gamma $.
Given that $\theta$ satisfies $x \overset{?}= f( t_1, \ldots, t_n) $ and it is enough to show that $\theta$ also satisfies $x \overset{?}=f( s_1, \ldots,s_n)$.
We write $x \theta \overset{?}= f( t_1, \ldots,t_n)\theta \overset{?}= f( t_1\theta, t_2\theta,\ldots,t_n\theta) \overset{?}= 
f( s_1\theta, s_2\theta,\ldots,s_n\theta)$ since $s_1 \theta \overset{?}= t_1 \theta ,\,\ldots,s_n \theta \overset{?}= t_n \theta$.
So, $\theta$ satisfies $x \overset{?}= f( t_1, \ldots,t_n)$ and $x \overset{?}=f( s_1, \ldots,s_n)$.
Hence, $\theta \models \{x \overset{?}=f( s_1, \ldots,s_n),\, x \overset{?}= f( t_1, \ldots,t_n)\}$.\\
\end{proof}
\begin{lemma} \label{lem:soundness2}
\emph{Let $\striple$ and $\striplep$ be two set triples such that \\$\striple {\scriptstyle \implies_{\Iach}} \bigvee_{i}(\Gamma_i||\triangle_i||\sigma_i)$ via AC unification. Let $\theta$ be a substitution such that $\theta \models \Gamma_i||\triangle_i||\sigma_i$. Then $\theta \models \Gamma ||\triangle||\sigma$.}
\end{lemma}
\begin{proof}
\textbf{AC Unification.}
$$\frac{\Psi \cup \Gamma || \triangle || \sigma}{GetEqs(\theta_{1}) \cup \Gamma|| \triangle || \sigma \vee\ldots \vee GetEqs(\theta_{n}) \cup \Gamma|| \triangle || \sigma }$$
Given that $\theta \models GetEqs(\theta_{1}) \cup \Gamma|| \triangle || \sigma \vee\ldots \vee GetEqs(\theta_{n}) \cup \Gamma|| \triangle || \sigma$.
This means that $\theta$ satisfies $GetEqs(\theta_{1}) \cup \Gamma|| \triangle || \sigma, \ldots, GetEqs(\theta_{n}) \cup \Gamma|| \triangle || \sigma$. Which implies that $\theta$ also satisfies $\Psi$. 
\end{proof}
By combining Lemma~\ref{lem:soundness1} and Lemma~\ref{lem:soundness2}, we have:
\begin{lemma}\label{lem:soundness3}
\emph{Let $\striple$ and $\striplep$ be two set triples such that \\$\striple {\scriptstyle \implies_{\Iach}} \bigvee_{i}(\Gamma_i||\triangle_i||\sigma_i)$. Let $\theta$ be a substitution such that $\theta \models \Gamma_i||\triangle_i||\sigma_i$. Then $\theta \models \Gamma ||\triangle||\sigma$.}
\end{lemma}
Then by induction on Lemma~\ref{lem:soundness3}, we get the following theorem:
\begin{theorem}\label{thm:soundness}
\emph{Let $\striple$ and $\striplep$ be two set triples such that \\$\striple \overset{*} {\scriptstyle \implies_{\Iach}} \bigvee_{i}(\Gamma_i||\triangle_i||\sigma_i$). Let $\theta$ be a substitution such that $\theta \models \Gamma_i||\triangle_i||\sigma_i$. Then $\theta \models \Gamma ||\triangle||\sigma$.}
\end{theorem}
We have the following corollary from Theorem~\ref{thm:soundness}:
\begin{theorem}[\textbf{Soundness}]
\emph{Let $\sigma$ be a set of equations. Suppose that we get $\bigvee_{i}(\Gamma_i||\triangle_i||\sigma_i)$ after exhaustively applying the rules from $\Iach$ to $\striple$, i.e, 
$\striple \overset{*} {\scriptstyle \implies_{\Iach}} \bigvee_{i}(\Gamma_i||\triangle_i||\sigma_i$), where for each $i$, no rules applicable to $\Gamma_i||\triangle_i||\sigma_i$. Let $\Sigma = \{\sigma_{i} \mid \Gamma_{i} = \emptyset\}$. Then any member of $\Sigma$ is an $ACh$-unifier of $\Gamma$.}
\end{theorem}
\subsection{Completeness}
Before going to prove the completeness of our inference system, we present a definition below:
\begin{definition}[Directed conservative extension]
\emph{Let $\bigvee_{i}(\Gamma_i||\triangle_i||\sigma_i)$ and $\bigvee_{i}(\Gamma'_i||\triangle'_i||\sigma'_i)$ be two set triples. $\bigvee_{i}(\Gamma'_i||\triangle'_i||\sigma'_i)$ is called a directed conservative extension of $\bigvee_{i}(\Gamma_i||\triangle_i||\sigma_i)$, if for any substitution $\theta$, such that $\theta \models \Gamma_i||\triangle_i||\sigma_i$, then there exists $k$ and $\sigma$, whose domain is the variables in $Var(\Gamma'_{k})\setminus Var(\Gamma_{k})$, such that $\theta\sigma \models \Gamma'_i||\triangle'_i||\sigma'_i$. If $\bigvee_{i}(\Gamma_i||\triangle_i||\sigma_i)$ (resp. $\bigvee_{i}(\Gamma'_i||\triangle'_i||\sigma'_i)$) only contains one set triple $\striple$ (resp. $\striplep$), we say $\bigvee_{i}(\Gamma'_i||\triangle'_i||\sigma'_i)$ (resp. $\striplep$) is a directed conservative extension of $\striple$.}
\end{definition}
Next, we show that our inference procedure never loses any solution.
\begin{lemma}\label{lem:comACU}
\emph{Let $\striple$ be a set triple. If there exists a set triple $\striplep$ such that $\striple \Rightarrow_{\Iach} \striplep$ via all the rules of $\Iach$ except AC unification, then $\striplep$ is a directed conservative extension of $\striple$. }
\end{lemma}
\begin{proof}
\textbf{Trivial.} It is trivially true.\\
 \noindent
\textbf{Occur Check.} In the homomorphism theory, no term can be equal to a subterm of itself.
This is because the number of $+$ symbols and h-depth of each variable stay the same with the application of the homomorphism equation $h(x_1+\cdots+x_n)\overset{?}= h(x_1)+\cdots+h(x_n)$.
So, the given problem has no solution in the homomorphism theory.\\
\textbf{Bound Check.}
 We see that there exists a variable $y$ with the h-depth $\kappa+1$ in the graph, that is, there is a variable $x$ above $y$ with $ \kappa+1$ h-symbols below it.
 Let $\theta$ be a solution of the unification problem $\Gamma$. Then the term $x\theta$ has the h-height $\kappa +1$, but the term $x\theta$ is also a subterm of some $s_i\theta$ or $t_i \theta$ in the original unification problem.
Hence, the unification problem $\Gamma$ has no solution within the given bound $\kappa$.\\
\ignore{Variable Elimination. It is trivially true.\\}
\textbf{Clash. }We don't have a rewrite rule that deals with
the uninterpreted function symbols, i.e., the function symbols which are not in $\{ h,\,+\}$. So the given problem has to have no solution.\\
\textbf{Splitting.} We have to make sure that we never lose any solution with this rule.
Here we consider the rewrite system $R_1$ which has the rewrite rule $h(x_1+\cdots+x_n) \rightarrow h(x_1)+\cdots+h(x_n)$.
 In order to apply this rule the term under the $h$ should be the sum of $n$ variables. 
The problem $\{h(y) \overset{?}= x_1 +\cdots+x_n\}$ is replaced by the set $\{h(v_1+\cdots+v_n) \overset{?}= x_1 +\cdots+x_n\}$ with
the substitution $\{y\mapsto v_1 + \cdots+v_n \}$. 
Then we have the equation with the reduced term in $R_1$ is the equation $h(v_1)+\cdots+h(v_n) \overset{?}= x_1 +\cdots+x_n$, and the substitution $\{y\mapsto v_1 +\cdots+ v_n, x_1 \mapsto h(v_1),\ldots, x_n \mapsto h(v_n)\}$. 
Hence, we never lose any solution here. \\
\textbf{Decomposition.}
If $f$ is the top symbol on both sides of an equation then there is no other rule to solve it except the Decomposition rule, where $f\neq h$ and $f\neq +$. So, we never lose any solution.\\
\par To cover the case where the top symbol is $h$ for the terms on both sides of an equation, we consider the rewrite system $R_2$ which has the rewrite rule $h(x_1)+h(x_2) \rightarrow h(x_1+x_2)$. In the homomorphism theory with the rewrite system $R_2$, we cannot reduce the term $h(t)$. 
So, we solve the equation of the form $h(t_1) \overset{?}=h(t_2)$ only with the Decomposition rule. Hence, we never lose any solution here too.
\end{proof}
\begin{lemma}\label{lem:ACU}
\emph{Let $\striple$ be a set triple. If there exists a set triple $\striplep$ such that $\striple {\scriptstyle \implies_{\Iach}} \bigvee_{i}(\Gamma_i||\triangle_i||\sigma_i)$ via AC unification, then $\bigvee_{i}(\Gamma_i||\triangle_i||\sigma_i)$ is a directed conservative extension of $\striple$. }
\end{lemma}
\begin{proof}
Since the buit-in AC unification algorithm is complete, we never lose any solutions on the application of this rule.
\ignore{Assume that there are equations in $\Psi$ that contains both $h$ and $+$ symbols. Then the built-in unification algorithm on $\Psi$ may loose solutions since $h$ is considered as an uninterpreted symbol. However, the missing solutions are regained on the other branch of the unification problem.}
\end{proof}
By combining Lemma~\ref{lem:comACU} and Lemma~\ref{lem:ACU}, we have:
\begin{lemma}\label{lem:1step}
\emph{Let $\striple$ be a set triple. If there exists a set of set triples $\bigvee_{i}(\Gamma_i||\triangle_i||\sigma_i)$ such that $\striple {\scriptstyle \implies_{\Iach}} \bigvee_{i}(\Gamma_i||\triangle_i||\sigma_i)$, then $\bigvee_{i}(\Gamma_i||\triangle_i||\sigma_i)$ is a directed conservative extension of $\striple$. }
\end{lemma}
By induction on Lemma~\ref{lem:1step}, we get:
\begin{theorem}\label{lem:oneormore}
\emph{Let $\striple$ be a set triple. If there exists a set of set triples $\bigvee_{i}(\Gamma_i||\triangle_i||\sigma_i)$ such that $\striple \overset{+} {\scriptstyle \implies_{\Iach}} \bigvee_{i}(\Gamma_i||\triangle_i||\sigma_i)$, then $\bigvee_{i}(\Gamma_i||\triangle_i||\sigma_i)$ is a directed conservative extension of $\striple$. }
\end{theorem}
We get the following corollary from the above theorem:
\begin{theorem}[\textbf{Completeness}] \label{thm:completeness}
\emph{Let $\Gamma$ be a set of equations. Suppose that we get $\bigvee_{i}(\Gamma_{i}||\triangle_{i}|| \sigma_{i})$ after applying the rules from $\Iach$ to $\striple$ exhaustively, that is, \\
$\striple \overset{*} {\scriptstyle \implies_{\Iach}} \bigvee_{i}(\Gamma_{i}||\triangle_{i}|| \sigma_{i})$, where for each $i$, none of the rules applicable on $\Gamma_{i}||\triangle_{i}|| \sigma_{i}$. Let $\Sigma = \{\sigma_{i} \mid \Gamma_{i} = \emptyset\}$. Then for any $ACh$-unifier $\theta$ of $\Gamma$, there exists a $\sigma \in \Sigma$, such that $\sigma \lesssim^{Var(\Gamma)} _{ACh}\theta$.}
\end{theorem}
\section{Implementation}
We have implemented the algorithm in the Maude programming language\footnote{\url{http://maude.cs.illinois.edu/w/index.php/The_Maude_System}}. The implementation of this inference system is available\footnote{\url{https://github.com/ajayeeralla/Unification_ACh}}. We chose the Maude language because it provides a nice environment for expressing inference rules of this algorithm. \ignore{and the implementation of this algorithm will be integrated into the Maude-NPA\footnote{\url{http://maude.cs.illinois.edu/w/index.php/Maude_Tools:_Maude-NPA}}, a protocol analyzer written in Maude and developed by Naval Research Laboratory (NRL), USA, at some time.}The system specifications of this implementation are Ubuntu 14.04 LTS, Intel Core i5 3.20 GHz, and 8 GiB RAM with Maude 2.6.\par
We give a table to show some of our results. In the given table, we use five columns: Unification problem, Real Time, time to terminate the program in ms (milliseconds), Solution either $\bot$ for no solution or Yes for solutions, \# Sol. for number of solutions, and Bound $\kappa$.
It makes sense that the real time keeps increasing as the given h-depth $\kappa$ increases for the first problem where the other problems give solutions, but in either case the program terminates.
\begin{table*}[!t]
 \footnotesize
 \begin{tabular} {|ccccc|}
\hline \hline
 Unification Problem & Real Time & Solution & \# Sol.& Bound \\ \hline
\hline
 $\{h(y) \overset{?}= y + x\}$ & 674ms & $\bot$ & 0 &10 \\ \hline
 $\{h(y) \overset{?}= y + x\}$ & 15880ms & $\bot$ & 0 &20 \\ \hline
 $\{h(y) \overset{?}= x_1 + x_2\}$ &5ms & Yes & 1 &10 \\ \hline
 $\{h(h(x)) \overset{?}= h(h(y))\}$ &2ms & Yes & 1 & 10 \\ \hline
 $\{x + y_1 \overset{?}= x + y_2\}$ &3ms & Yes &1 &10 \\ \hline
 $\{v \overset{?}= x + y, v \overset{?}= w+z , s \overset{?}= h(t)\}$ &46ms & Yes & 10 &10 \\ \hline
 $\{v \overset{?}= x_1 + x_2 , v \overset{?}= x_3 + x_4 , x_1 \overset{?}= h(y), x_2 \overset{?}= h(y)\}$ &100ms & Yes & 6 &10 \\ \hline
 $\{h(h(x)) \overset{?}= v+w+y + z\}$ &224ms & Yes & 1&10 \\ \hline
 $\{v \overset{?}= (h(x)+y),v\overset{?}= w+z\}$ &55ms & Yes & 7&10 \\ \hline
  $\{f(x , y) \overset{?}= h(x_1)\}$ &0ms & $\bot$ & 0&10 \\ \hline
  $\{f(x_1 , y_1) \overset{?}= f(x_2 , y_2)\}$ &1ms & Yes & 1&10 \\ \hline
  $\{v \overset{?}= x_1 + x_2 , v \overset{?}= x_3 + x_4\} $ &17ms & Yes & 7 &10 \\ \hline
  $\{f(x_1 , y_1) \overset{?}= g(x_2 , y_2)\}$ &0ms & $\bot$ & 0&10 \\ \hline
  $\{h(y) \overset{?}= x , y \overset{?}= h(x)\}$ &0ms & $\bot$ & 0&10 \\ \hline
  \end{tabular}
  \caption{\em{Tested results with bounded ACh-unification algorithm}}
  \label{tab:imp}
 \end{table*}
\section{Conclusion}

We introduced a set of inference rules to solve the unification problem modulo the homomorphism theory $h$ over an AC symbol $+$, by enforcing a threshold $\kappa$ on the h-depth of any variable.
\ignore{Our algorithm finds all the solutions or unifiers within the given h-depth $\kappa$. We implemented the algorithm in Maude because the inference rules are easy to write in Maude, and also we hope to incorporate them in the Maude-NPA tool, a protocol analyzer written in Maude. Our work on this topic actually came out of work on the Maude-NPA tool.} Homomorphism is a property that is very common in cryptographic algorithms. So, it is important to analyze cryptographic protocols in the homomorphism theory. Some of the algorithms and details in this direction can be seen in~\cite{ALLNR2, EKLMMNS, ALLNR1}. However, none of those results perform ACh unification because that is undecidable. \ignore{, One way around this, is to assume that identity and an inverse exist, but because of the way the Maude-NPA works, it would still be necessary to unify modulo ACh. So a unification algorithm there becomes crucial.} We believe that our approximation is a good way to deal with it. We also tested some problems and the results are shown in Table~\ref{tab:imp}.\\
\textbf{Acknowledgments.} We thank anonymous reviewers who have provided useful comments. Ajay Kumar Eeralla was partially supported by NSF Grant CNS 1314338.



\end{document}